\documentclass[a4paper,UKenglish,cleveref, autoref, thm-restate,authorcolumns,pagebackref]{lipics-v2019}
\usepackage[utf8]{inputenc}
\usepackage{amsmath}
\usepackage{amsfonts}
\usepackage{amssymb}
\usepackage{amsthm}
\usepackage{mathtools}
\usepackage{subcaption}
\usepackage{paralist}
\usepackage[noend,ruled,algo2e]{algorithm2e}
\usepackage{xspace}
\usepackage{tabularx}
\usepackage{booktabs}
\usepackage{tikz}
\usetikzlibrary{decorations.pathreplacing}

\usepackage[sort,numbers]{natbib}

\usepackage[backgroundcolor=gray!10]{todonotes}

\newcommand{\BDTW}{\textsc{BDTW-Mean}\xspace}
\newcommand{\MSS}{\textsc{MSS}\xspace}
\newcommand{\Q}{\mathbb{Q}}

\DeclareMathOperator{\dtw}{dtw}

\newtheorem{observation}[theorem]{Observation}

\crefname{question}{Question}{Questions}

\newcommand{\problemdef}[3]{
	\begin{center}
  \begin{minipage}{0.95\textwidth}
    \noindent
    \normalsize
    \textsc{#1}

			\vspace{2pt}
			\setlength{\tabcolsep}{3pt}
			\begin{tabularx}{\textwidth}{@{}lX@{}}
					\normalsize\textbf{Input:} 		& \normalsize #2 \\
					\normalsize\textbf{Question:} 	& \normalsize #3
				\end{tabularx}
  \end{minipage}
	\end{center}
      }

\newcommand{\optproblemdef}[3]{
	\begin{center}
  \begin{minipage}{0.95\textwidth}
    \noindent
    \normalsize
    \textsc{#1}

			\vspace{2pt}
			\setlength{\tabcolsep}{3pt}
			\begin{tabularx}{\textwidth}{@{}lX@{}}
					\normalsize\textbf{Input:} 		& \normalsize #2 \\
					\normalsize\textbf{Task:} 	& \normalsize #3
				\end{tabularx}
  \end{minipage}
	\end{center}
      }

\title{Faster Binary Mean Computation Under Dynamic Time Warping}

\titlerunning{Fast Binary DTW-Mean} 

\author{Nathan Schaar}{Technische Universit\"at Berlin, Faculty~IV, 
Algorithmics and Computational Complexity }{n.schaar@campus.tu-berlin.de}{}{Partially supported by DFG NI~369/19.}

\author{Vincent Froese}{Technische Universit\"at Berlin, Faculty~IV, 
Algorithmics and Computational Complexity}{vincent.froese@tu-berlin.de}{}{}

\author{Rolf Niedermeier}{Technische Universit\"at Berlin, Faculty~IV, 
Algorithmics and Computational Complexity}{rolf.niedermeier@tu-berlin.de}{}{}

\authorrunning{N.~Schaar, V.~Froese and R.~Niedermeier} 

\Copyright{Nathan Schaar, Vincent Froese and Rolf Niedermeier} 



\funding{}

\nolinenumbers 

\hideLIPIcs  

\EventEditors{John Q. Open and Joan R. Access}
\EventNoEds{2}
\EventLongTitle{42nd Conference on Very Important Topics (CVIT 2016)}
\EventShortTitle{CVIT 2016}
\EventAcronym{CVIT}
\EventYear{2016}
\EventDate{December 24--27, 2016}
\EventLocation{Little Whinging, United Kingdom}
\EventLogo{}
\SeriesVolume{42}
\ArticleNo{23}

\begin{document}

\maketitle

\begin{abstract} 
Many consensus string problems are based on Hamming distance.
We replace Hamming distance by the more flexible (e.g., easily 
coping with different input string lengths) dynamic time warping 
distance, best known from applications in time series mining.
Doing so, we study the problem of finding a mean string 
that minimizes the sum of (squared) dynamic time warping distances 
to a given set of input strings.
While this problem is known to be NP-hard (even for strings over a three-element alphabet), we address the binary alphabet case which is known to be polynomial-time solvable.
We significantly improve on a previously known algorithm in terms of 
worst-case running time.
Moreover, we also show the practical usefulness of one of our algorithms 
in experiments with real-world and synthetic data.
Finally, we identify special cases solvable in linear time
(e.g., finding a mean of only two binary input strings) 
and report some empirical findings concerning combinatorial
properties of optimal means.
\end{abstract}

\newpage

\section{Introduction}
Consensus problems are an integral part of stringology. 
For instance, in the frequently studied \textsc{Closest String} 
problem one is given $k$~strings 
of equal length and the task is to find a center string that minimizes 
the maximum Hamming distance to all $k$~input 
strings. 
\textsc{Closest String} is NP-hard even for binary alphabet~\cite{FL97} 
and has been extensively studied in context of approximation and 
parameterized algorithmics~\cite{BHKN14,CW11,CMW12,CMW16,GNR03,LMW02,NS12,YCMW19}. 
Notably, when one wants to minimize the sum of distances instead of 
the maximum distance, the problem is easily solvable in linear 
time by taking at each position a letter that appears
most frequently in the input strings.

Hamming distance, however, is quite limited in many application contexts; for instance, how to 
define a center string in case of input strings that do not all have the 
same length?  In context of analyzing time series (basically strings where 
the alphabet consists of rational numbers), the  ``more flexible''
\emph{dynamic time warping distance}~\cite{SC78} enjoys high 
popularity and can be computed for two input strings in subquadratic 
time~\cite{GS18,Kuszmaul19}, essentially matching corresponding conditional 
lower bounds~\cite{ABW15,BK15}. Roughly speaking (see Section~\ref{sec:prelim}
for formal definitions and an example),
measuring the dynamic time warping distance (dtw for short) can be seen 
as a two-step process: First, one
aligns one time series with the other (by stretching them via duplication of elements) such that both time series end up with the same length.
Second, one then calculates the Euclidean distance of the aligned
time series (recall that here the alphabet consists 
of numbers). Importantly, restricting to the binary case,
the dtw~distance  of two time series can be 
computed in $O(n^{1.87})$ time~\cite{ABW15}, where $n$ is the maximum time series length (a result that will also be relevant for our work).

With the dtw~distance at hand, the most fundamental consensus 
problem in this (time series) context is, given $k$~input ``strings'' 
(over rational numbers), 
compute a mean string that minimizes the sum of (squared) dtw~distances 
to all input strings. This problem is known as 
\textsc{DTW-Mean} in the literature and only recently 
has been shown to be NP-hard~\cite{BDS19,BFN18}. For the most basic
case, namely binary alphabet (that is, input and output are binary), however, the problem is 
known to be solvable in $O(kn^3)$~time~\cite{BFFJNS19}.
By way of contrast, if one allows the mean to contain any rational 
numbers, then the problem is NP-hard even for binary inputs~\cite{BFN18}.
Moreover, the problem is also NP-hard for ternary input and output~\cite{BDS19}.

Formally, in this work we study the following problem:
\problemdef{Binary DTW-Mean (\BDTW)}
{Binary strings~$s_1,\ldots,s_k$ of length at most~$n$ and~$c\in\Q$.}
{Is there a binary string~$z$ such that~$F(z)\coloneqq\sum_{i=1}^k\dtw(s_i,z)^2\le c$?}
Herein, the $\dtw$ function is formally defined in Section~\ref{sec:prelim}.
The study of the special case of binary data may occur when one deals with 
binary states (e.g., switching between the active and the inactive mode of 
a sensor); binary data were recently studied in the dynamic time warping 
context~\cite{MCAHM16,SDHDKJC18}. Clearly, binary data can always be generated 
from more general data by ``thresholding''.

Our main theoretical result is to show that \BDTW{} can 
be solved in $O(kn^{1.87})$ and $O(k(n+m(m-\mu)))$ time, respectively, where $m$ is the maximum and~$\mu$ is the median condensation length of the input strings (the condensation of a string is obtained by repeatedly removing one of two identical consecutive elements).
While the first algorithm, relies on an intricate ``blackbox-algorithm'' for a
certain number problem from the literature (which so far was never implemented), the second algorithm (which we implemented) is more directly based on combinatorial arguments. 
Anyway, our new bounds improve on the standard 
$O(kn^3)$-time bound~\cite{BFFJNS19}. Moreover, we also experimentally 
tested our second algorithm and compared it to the standard one, clearly 
outperforming it (typically by orders of magnitude) on real-world 
and on synthetic instances.  
Further theoretical results comprise linear-time algorithms for 
special cases (two input strings or three input strings with some 
additional constraints). 
Further empirical results relate to the typical shape of a mean.

\section{Preliminaries}\label{sec:prelim}


For~$n\in\mathbb{N}$, let~$[n] := \{1,\ldots,n\}$.
We consider binary strings~$x=x[1]x[2]\ldots x[n]\in\{0,1\}^n$.
We denote the length of~$x$ by~$|x|$ and we also denote the last symbol~$x[n]$ of~$x$ by~$x[-1]$.
For~$1\le i \le j \le |x|$, we define the substring~$x[i,j]:=x[i]\ldots x[j]$.
A maximal substring of consecutive 1's (0's) in~$x$ is called a \emph{1-block} (\emph{0-block}).
The $i$-th block of~$x$ (from left to right) is denoted~$x^{(i)}$.
A string $x$ is called \emph{condensed} if no two consecutive elements are equal, that is, every block is of size~1.
The \emph{condensation} of~$x$ is denoted~$\tilde{x}$ and is defined as the condensed string obtained by removing one of two equal consecutive elements of~$x$ until the remaining series is condensed. Note that the condensation length~$|\tilde{x}|$ equals the number of blocks in~$x$.

The dynamic time warping distance measures the similarity of two strings using non-linear alignments defined via so-called warping paths.
\begin{definition}
  A \emph{warping path} of order~$m\times n$ is a sequence~$p=(p_1,\ldots,p_L)$, $L\in\mathbb{N}$,
  of index pairs $p_\ell=(i_\ell,j_\ell)\in [m]\times[n]$, $1\le \ell \le L$, such that
  \begin{compactenum}[(i)]
    \item $p_1=(1,1)$,
    \item $p_L=(m,n)$, and
    \item $(i_{\ell+1}-i_\ell, j_{\ell+1}-j_\ell)\in \{(1,0),(0,1),(1,1)\}$ for each~$\ell \in [L-1]$.
  \end{compactenum}
\end{definition}

A warping path can be visualized within an $m\times n$ ``warping matrix'' (see \Cref{fig:optwarpbinary}).
The set of all warping paths of order~$m\times n$ is denoted by~$\mathcal{P}_{m,n}$.
A warping path~$p\in\mathcal{P}_{m,n}$ defines an \emph{alignment} between two strings~$x\in \Q^m$ and~$y\in\Q^n$ in the following way:
A pair~$(i,j)\in p$ \emph{aligns} element~$x_i$ with~$y_j$ with a local cost of~$(x_i-y_j)^2$.
The dtw~distance between two strings~$x$ and~$y$ is defined as
$$\dtw(x,y) := \min_{p\in\mathcal{P}_{m,n}}\sqrt{\sum_{(i,j)\in p}(x_i-y_j)^2}.$$
It is computable via standard dynamic programming in~$O(mn)$ time\footnote{Throughout this work, we assume that all arithmetic operations can be carried out in constant time.}~\cite{SC78}, with recent theoretical improvements to subquadratic time~\cite{GS18,Kuszmaul19}.

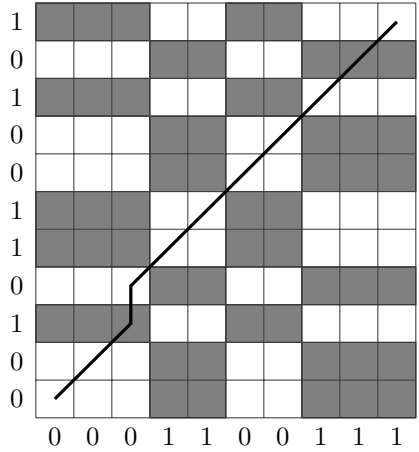
\begin{figure}[t]
    \centering
    \begin{tikzpicture}[scale=.5]

      \draw[fill=black!50] (3,0) rectangle (5,2);
      \draw[fill=black!50] (7,0) rectangle (10,2);

      \draw[fill=black!50] (3,3) rectangle (5,4);
      \draw[fill=black!50] (7,3) rectangle (10,4);
      
      \draw[fill=black!50] (3,6) rectangle (5,8);
      \draw[fill=black!50] (7,6) rectangle (10,8);

      \draw[fill=black!50] (3,9) rectangle (5,10);
      \draw[fill=black!50] (7,9) rectangle (10,10);

      \draw[fill=black!50] (0,2) rectangle (3,3);
      \draw[fill=black!50] (5,2) rectangle (7,3);

      \draw[fill=black!50] (0,4) rectangle (3,6);
      \draw[fill=black!50] (5,4) rectangle (7,6);
      
      \draw[fill=black!50] (0,8) rectangle (3,9);
      \draw[fill=black!50] (5,8) rectangle (7,9);

      \draw[fill=black!50] (0,10) rectangle (3,11);
      \draw[fill=black!50] (5,10) rectangle (7,11);
  
      \draw[black!75] (0,0) grid (10,11);
      \draw[very thick] (0.5,0.5) -- (2.5,2.5) -- (2.5,3.5) -- (9.5,10.5);
      \node at (0.5,-0.5) {0};
      \node at (1.5,-0.5) {0};
      \node at (2.5,-0.5) {0};
      \node at (3.5,-0.5) {1};
      \node at (4.5,-0.5) {1};
      \node at (5.5,-0.5) {0};
      \node at (6.5,-0.5) {0};
      \node at (7.5,-0.5) {1};
      \node at (8.5,-0.5) {1};
      \node at (9.5,-0.5) {1};

      \node at (-0.5,0.5) {0};
      \node at (-0.5,1.5) {0};
      \node at (-0.5,2.5) {1};
      \node at (-0.5,3.5) {0};
      \node at (-0.5,4.5) {1};
      \node at (-0.5,5.5) {1};
      \node at (-0.5,6.5) {0};
      \node at (-0.5,7.5) {0};
      \node at (-0.5,8.5) {1};
      \node at (-0.5,9.5) {0};
      \node at (-0.5,10.5) {1};
      
    \end{tikzpicture}
    \caption{An optimal warping path for the strings~$x=00101100101$ (vertical axis) and $y=0001100111$ (horizontal axis).
      Black cells have local cost 1.
      The string~$x$ consists of eight blocks with sizes 2,1,1,2,2,1,1,1 and~$y$ consists of four blocks with sizes 3,2,2,3. An optimal warping path has to pass through $(8-4)/2 =2$ non-neighboring blocks of the six inner blocks of~$x$.
    }
    \label{fig:optwarpbinary}
  \end{figure}

\section{DTW on Binary Strings}\label{sec:bin-dtw}

We briefly discuss some known results about the dtw~distance between binary strings since these will be crucial for our algorithms for \BDTW.

\citet[Section~5]{ABW15} showed that the dtw~distance of two binary strings of length at most~$n$ can be computed in~$O(n^{1.87})$ time.
They obtained this result by reducing the dtw~distance computation to the following integer problem.

\optproblemdef{Min 1-Separated Sum (\MSS)}
{A sequence $(b_1,\ldots,b_m)$ of~$m$ positive integers and an integer $r\ge 0$.}
{Select $r$ integers~$b_{i_1},\ldots,b_{i_r}$ with $1 \le i_1 < i_2 < \dots < i_r \le m$ and
  $i_j < i_{j+1}-1$ for all~$1\le j < r$ such that $\sum_{j=1}^rb_{i_j}$ is minimized.}
\noindent
The integers of the \MSS instance correspond to the block sizes of the input string which contains more blocks.

\begin{theorem}[{\cite[Theorem~8]{ABW15}}]\label{thm:binDTW}
  Let~$x\in\{0,1\}^m$ and $y\in\{0,1\}^n$ be two binary strings such that~$x[1]=y[1]$, $x[m]=y[n]$, and $|\tilde{x}| \ge |\tilde{y}|$.
  Then, $\dtw(x,y)^2$ equals the sum of a solution for the \MSS instance $\left((|x^{(2)}|,\ldots,|x^{(|\tilde{x}|-1)}|),(|\tilde{x}|-|\tilde{y}|)/2\right)$.
\end{theorem}

The idea behind \Cref{thm:binDTW} is that exactly $(m'-n')/2$ non-neighboring blocks of~$x$ are misaligned in any warping path (note that $m'-n'$ is even since~$x$ and~$y$ start and end with the same symbol).
An optimal warping path can thus be obtained from minimizing the sum of block sizes of these misaligned blocks. For example, in \Cref{fig:optwarpbinary} the dtw~distance corresponds to a solution of the \MSS instance~$((1,1,2,2,1,1),2)$.

\citet[Theorem~10]{ABW15} showed how to solve~\MSS in~$O(n^{1.87})$ time, where~$n = \sum_{i=1}^mb_i$.
They gave a recursive algorithm that, on input~$\left((b_1,\ldots,b_{m}),r\right)$, outputs four lists~$C^{00}, C^{0*}, C^{*0}$, and $C^{**}$, where, for~$t\in\{0,\ldots,r\}$,
\begin{itemize}
\item $C^{**}[t]$ is the sum of a solution for the \MSS instance $\left((b_1,\ldots,b_{m}),t\right)$,
\item $C^{0*}[t]$ is the sum of a solution for the \MSS instance $\left((b_2,\ldots,b_{m}),t\right)$,
\item $C^{*0}[t]$ is the sum of a solution for the \MSS instance $\left((b_1,\ldots,b_{m-1}),t\right)$, and
\item $C^{00}[t]$ is the sum of a solution for the \MSS instance $\left((b_2,\ldots,b_{m-2}),t\right)$.
\end{itemize}

\noindent
Note that~$C^{**}[r]$ yields the solution.
We will make use of their algorithm when solving \BDTW.
We will also use the following simple dynamic programming algorithm for \MSS which is faster for large input integers.

\begin{lemma}\label{lem:MSSdp}
  \textsc{Min 1-Separated Sum} is solvable in~$O(mr)$ time.
\end{lemma}

\begin{proof}
  Let~$\left((b_1,\ldots,b_m),r\right)$ be an \MSS instance.
  We define a dynamic programming table~$M$ as follows:
  For each $i\in [m]$ and each~$j\in\{0,\ldots,\min(r,\lceil i/2 \rceil)\}$, $M[i,j]$ is the sum of a solution of the subinstance~$\left((b_1,\ldots,b_i),j\right)$.
  Clearly, it holds $M[i,0]=0$ and~$M[i,1]=\min\{b_1,\ldots,b_i\}$ for all~$i$.
  Further, it holds $M[3,2]=b_1+b_3$.
  For all~$i\in\{4,\ldots,m\}$ and~$j\in\{2,\ldots,\min(r,\lceil i/2\rceil)\}$, the following recursion holds
  \[M[i,j] = \min(b_i + M[i-2,j-1], M[i-1,j]).\]
  Hence, the table~$M$ can be computed in~$O(mr)$ time.
\end{proof}

Note that the above algorithms only compute the dtw~distance between binary strings with equal starting and ending symbol.
However, it is an easy observation that the dtw~distance of arbitrary binary strings can recursively be obtained from this via case distinction on which first and/or which last block to misalign. 

\begin{observation}[{\cite[Claim~6]{ABW15}}]\label{obs:wlog}
  Let~$x\in\{0,1\}^m$, $y\in\{0,1\}^n$ with $m':=|\tilde{x}|\ge n':=|\tilde{y}|$.
  Further, let~$a:=|x^{(1)}|$, $a':=|x^{(m')}|$, $b:=|y^{(1)}|$, and  $b':=|y^{(n')}|$.
  The following holds:
  \begin{compactitem}
    \item If $x[1]\neq y[1]$, then
      \[\dtw(x,y)^2 = \begin{cases}
          \max(a,b), &m'=n'=1\\
          a +\dtw(x[a+1,m],y)^2, &m' > n'=1\\
          \min\left(a+\dtw(x[a+1,m],y)^2, b+\dtw(x,y[b+1,n])^2\right), &n'>1
        \end{cases}
          .\]
    \item If $x[1]=y[1]$ and $x[m]\neq y[n]$, then
      \[\dtw(x,y)^2 = \begin{cases}
          a'+\dtw(x[1,m-a'],y)^2, &n'=1\\
          \min\left(a'+\dtw(x[1,m-a'],y)^2, b'+\dtw(x,y[1,n-b'])^2\right), &n'>1
        \end{cases}.\]
  \end{compactitem}
\end{observation}

For condensed strings, \citet{BFFJNS19} derived the following useful closed form for the dtw~distance (which basically follows from \Cref{obs:wlog,thm:binDTW}).

\begin{lemma}[{\cite[Lemma~1 and~2]{BFFJNS19}}]\label{lem:dtwdist}
  For a condensed binary string~$x$ and a binary string~$y$ with~$|\tilde{y}| \le |x|$, it holds that
  \[\dtw(x,y)^2 =
    \begin{cases}
      \lceil(|x| - |\tilde{y}|)/2\rceil, & x_1=y_1\\
      2, & x_1\neq y_1 \wedge |x|=|\tilde{y}|\\
      1+\lfloor (|x| - |\tilde{y}|)/2\rfloor, & x_1\neq y_1 \wedge |x|>|\tilde{y}|
    \end{cases}.\]
\end{lemma}
\noindent
According to~\Cref{lem:dtwdist}, one can compute the dtw~distance in constant time when the condensation lengths of the inputs are known and the string with longer condensation length is condensed. 

Our key lemma now states that the dtw~distances between an arbitrary fixed string and all condensed strings of shorter condensation length can also be computed efficiently.

\begin{lemma}\label{lem:DTWdp}
  Let~$s\in\{0,1\}^n$ with $\ell\coloneqq |\tilde{s}|$.
  Given~$\ell$ and the block sizes $b_1,\ldots,b_\ell$ of~$s$, the dtw~distances between~$s$ and all condensed strings of lengths~$\ell', \ldots, \ell$ for some given $\ell'\le \ell$ can be computed in
  \begin{compactenum}[(i)]
    \item $O(n^{1.87})$ time and in
    \item $O(\ell (\ell -\ell'))$ time, respectively.
  \end{compactenum}
\end{lemma}

\begin{proof}
  Let~$x$ be a condensed string of length~$i\in\{\ell',\ldots,\ell\}$.
  \Cref{obs:wlog,thm:binDTW} imply that we essentially have to solve \MSS on four different subsequences of block sizes of~$s$ (depending on the first and last symbol of~$x$) in order to compute~$\dtw(s,x)$.  
  Namely, the four cases are $(b_2,\ldots,b_{\ell-1})$, $(b_3,\ldots,b_{\ell-1})$, $(b_2,\ldots,b_{\ell-2})$, and $(b_3,\ldots,b_{\ell-2})$. Let~$r\coloneqq \lceil(\ell-\ell')/2\rceil$

  $(i)$ We run the algorithm of \citet[Theorem~10]{ABW15} on the instance $\left((b_2,\ldots,b_{\ell-1}),r\right)$ to obtain in~$O(n^{1.87})$ time the four
  lists~$C^{\alpha\beta}$, for $\alpha,\beta\in\{0,1\}$, where~$C^{\alpha\beta}$ contains the solutions of~$\left((b_{2+\alpha},\ldots,b_{\ell-1-\beta}),r'\right)$ for all~$r'\in\{0,\ldots,r\}$.
  From these four lists, we can compute the requested dtw~distances (using \Cref{obs:wlog}) in~$O(\ell)$ time.
  
  $(ii)$ We compute the solutions of the four above \MSS instances using \Cref{lem:MSSdp}.
  For each $\alpha,\beta\in\{0,1\}$, let~$M^{\alpha\beta}$ be the dynamic programming table computed in~$O(\ell(\ell-\ell'))$ time for the instance~$\left((b_{2+\alpha},\ldots,b_{\ell-1-\beta}), r\right)$.
  Again, we can compute the requested dtw~distances from these four tables in~$O(\ell)$ time (using \Cref{obs:wlog}).
\end{proof}

\section{More Efficient Solution of \BDTW}\label{sec:main}

\citet{BFFJNS19} gave an $O(kn^3)$-time algorithm for \BDTW.
The result is based on showing that there always exists a condensed mean of length at most~$n+1$.
Thus, there are~$2(n+1)$ candidate strings to check. For each candidate, one can compute the dtw~distance to every input string in~$O(kn^2)$ time.
It is actually enough to only compute the dtw~distance for the two length-$(n+1)$ candidates to all~$k$ input strings since the resulting dynamic programming tables also yield all the distances to shorter candidates. That is, the running time can actually be bounded in~$O(kn^2)$.

We now give an improved algorithm.
To this end, we first show the following improved bounds on the (condensation) length of a mean.

\begin{lemma}\label{lem:bounds}
  Let~$s_1,\ldots,s_k$ be binary strings with~$|\tilde{s}_1| \le \dots \le |\tilde{s}_k|$ and let~$z$
  be a mean. Then, it holds~$\mu-2 \le |\tilde{z}| \le m+1$, where $\mu=|\tilde{s}_{\lceil k/2 \rceil}|$ is the median condensation length and~$m = |\tilde{s}_k|$ is the maximum condensation length.
\end{lemma}

\begin{proof}
  It suffices to show the claimed bounds for condensed means.
  Since $\dtw(\tilde{x},y)\le \dtw(x,y)$ holds for all strings~$x$, $y$~\cite[Proposition~1]{BFFJNS19}, the bounds also hold for arbitrary means.
  
  The upper bound~$m+1$ can be derived from \Cref{lem:dtwdist}.
  Let~$x$ be a condensed string of length $|x|\ge m+2$ and let~$x'\coloneqq x[1,m]$.
  If~$|x| > m+2$, then~$\dtw(x',s_i)^2 < \dtw(x,s_i)^2$ holds for every~$i\in[k]$, which implies~$F(x') = \sum_{i=1}^k\dtw(s_i,x')^2 < \sum_{i=1}^k\dtw(s_i,x)^2=F(x)$. Hence,~$x$ is not a mean.
  If~$|x| = m+2$, then $\dtw(x',s_i)^2 \le \dtw(x,s_i)^2$ holds for every~$i\in[k]$, that is, $F(x')\le F(x)$.
  If $F(x') < F(x)$, then~$x$ is clearly not a mean.
  If~$F(x') = F(x)$, then $\dtw(x',s_i)^2 = \dtw(x,s_i)^2$ holds for all~$i\in[k]$.
  In fact, $\dtw(x',s_i)^2 = \dtw(x,s_i)^2$ only holds if~$|\tilde{s}_i|=m$ and~$s_i[1]\neq x[1]$, in which case $\dtw(x,s_i)^2 = 2$.
  Thus, we have~$F(x)=2k$ and $\tilde{s}_1=\tilde{s}_2=\dots = \tilde{s}_k$. But then $\tilde{s}_1$ is clearly the unique mean (with~$F(\tilde{s}_1)=0$).
  
  For the lower bound, let~$x$ be a condensed string of length~$\ell < \mu-2$ and let~$x'\coloneqq x[1]\ldots x[\ell]x[\ell-1]x[\ell]$.
  Then, for every $s_i$ with~$|\tilde{s}_i| \le \ell$ (of which there are less than $\lceil k/2\rceil$ since $\ell < \mu$), it holds $\dtw(x',s_i)^2 \le \dtw(x,s_i)^2 +1$ (by \Cref{lem:dtwdist}).

  Now, for every~$s_i$ with~$|\tilde{s}_i| > \ell+2$ (of which there are at least~$\lceil k/2 \rceil$ since $\ell+2 < \mu$),
  it holds $\dtw(x',s_i)^2 \le \dtw(x,s_i)^2 -1$. This is easy to see from \Cref{thm:binDTW} for the case that~$s_i[1]=x'[1]$ and~$s_i[-1]=x'[-1]$ holds since the number of misaligned blocks of~$s_i$ decreases by at least one. From this, \Cref{obs:wlog} yields the other three possible cases of starting and ending symbols since the sizes of the first and last block of~$x$ and of~$x'$ are clearly all the same (one).

  It remains to consider input strings $s_i$ with~$\ell < |\tilde{s}_i| \le \ell +2$.
  We show that in this case~$\dtw(x',s_i)^2 \le \dtw(x,s_i)^2$ holds.
  Let~$|\tilde{s}_i|=\ell+2$. Note that then either~$x'[1]=s_i[1]$ and~$x'[-1]=s_i[-1]$ holds or $x'[1]\neq s_i[1]$ and~$x'[-1]\neq s_i[-1]$ holds.
  In the former case, it clearly holds~$\dtw(x',s_i)^2=0$ by \Cref{lem:dtwdist}.
  In the latter case, we clearly have~$\dtw(x,s_i)^2\ge 2$, and, by \Cref{lem:dtwdist}, we have $\dtw(x',s_i)^2=2$.
  Finally, let~$|\tilde{s_i}|=\ell+1$ and note that then either~$x'[1]=s_i[1]$ and~$x'[-1]\neq s_i[-1]$ holds or $x'[1]\neq s_i[1]$ and~$x'[-1]= s_i[-1]$ holds.
  Thus, we clearly have $\dtw(x,s_i)^2\ge 1$. By \Cref{lem:dtwdist}, we have $\dtw(x',s_i)^2 =1$.

  Summing up, we obtain~$F(x') \le F(x) + a - b$, where~$a=|\{i\in[k]\mid |\tilde{s}_i| < \ell\}| < \lceil k/2 \rceil$ and $b=|\{i\in[k]\mid |\tilde{s}_i| > \ell+2\}| \ge \lceil k/2 \rceil$.
  That is, $F(x') < F(x)$ and~$x$ is not a mean.
\end{proof}
\noindent
Note that the length bounds in \Cref{lem:bounds} are tight. For the upper bound, consider the two strings 000 and 111 having the two means 01 and 10. For the lower bound, consider the seven strings
0, 0, 0, 101, 101, 010, 010 with the unique mean~0.

\Cref{lem:bounds} upper-bounds the number of mean candidates we have to consider in terms of the condensation lengths of the inputs.
In order to compute the dtw~distances between mean candidates and input strings, we
can now use \Cref{lem:DTWdp}.
We arrive at the following result.

\begin{theorem}\label{thm:main}
  Let~$s_1,\ldots,s_k$ be binary strings with~$|\tilde{s}_1| \le \dots \le |\tilde{s}_k|$ and let $n=\max_{j=1,\ldots,k}|s_j|$, $\mu=|\tilde{s}_{\lceil k/2 \rceil}|$, and~$m = |\tilde{s}_k|$.
  The condensed means of~$s_1,\ldots,s_k$ can be computed in
  \begin{compactenum}[(i)]
    \item $O(kn^{1.87})$ time and in
    \item $O(k(n+m(m-\mu)))$ time.
  \end{compactenum}
\end{theorem}

\begin{proof}
  From \Cref{lem:bounds}, we know that there are~$O(m-\mu)$ many candidate strings to check.
  First, in linear time, we determine the block lengths for each~$s_j$.
  Now, let~$x$ be a candidate string, that is,~$x$ is a condensed binary string with~$\mu-2 \le |x| \le m+1$. We need to compute~$\dtw(x,s_j)^2$ for each~$j=1,\ldots,k$.
  Consider a fixed string~$s_j$.
  For all candidates~$x$ with$|x| \ge |\tilde{s}_j|$, we can simply compute $\dtw(x,s_j)^2$ in constant time using \Cref{lem:dtwdist}. For all~$x$ with~$|x| < |\tilde{s}_j|$, we can use \Cref{lem:DTWdp}. Thus, overall, we can compute the dtw~distances between all candidates and all input strings in~$O(kn^{1.87})$ time, or alternatively in~$O(km(m-\mu))$~time.
  Finally, we determine the candidates with the minimum sum of dtw~distances in~$O(k(m-\mu))$~time.
\end{proof}

We remark that similar results also hold for the related problems \textsc{Weighted Binary DTW-Mean},
where the objective is to minimize~$F(z):=\sum_{i=1}^kw_i\dtw(s_i,z)^2$ for some $w_i\ge 0$,
and \textsc{Binary DTW-Center} with $F(z):=\max_{i=1,\ldots,k}\dtw(s_i,z)^2$ (that is, the dtw version of \textsc{Closest String}).
It is easy to see that also in these cases there exists a condensed solution.
Moreover, the length is clearly bounded between the minimum and the maximum condensation length of the inputs. Hence, analogously to \Cref{thm:main}, we obtain the following.

\begin{corollary}
  \textsc{Weighted Binary DTW-Mean} and \textsc{Binary DTW-Center} can be solved in
  $O(kn^{1.87})$ time and in  $O(k(n+m(m-\nu)))$ time,
    where $m$ is the maximum condensation length and~$\nu$ is the minimum condensation length.
\end{corollary}

\section{Linear-Time Solvable Special Cases}

Notably, \Cref{thm:main}~(ii) yields a linear-time algorithm when~$m-\mu$ is constant
and also when all input strings have the same length~$n$ and~$m(m-\mu)\in O(n)$.
Now, we show two more linear-time solvable cases.

\begin{theorem}\label{thm:k=2}
  A condensed mean of two binary strings can be computed in linear time.
\end{theorem}

\begin{proof}
  Let~$s_1,s_2\in\{0,1\}^*$ be two input strings.
  We first determine the condensations and block sizes of~$s_1$ and $s_2$ in linear time.
  Let~$\ell_i\coloneqq|\tilde{s}_i|$, for $i\in[2]$, and assume that $\ell_1 \le \ell_2$.
  In the following, all claimed relations between dtw~distances can easily be verified using \Cref{obs:wlog} (together with \Cref{thm:binDTW}) and \Cref{lem:dtwdist}.
 
  If $\ell_1 = \ell_2$, then, by \Cref{thm:main} (with $\mu=m=\ell_1$), all condensed means can be computed in~$O(\ell_1)$ time.
  
  If~$\ell_1 < \ell_2$, then~$\tilde{s}_2$ is a mean.
  To see this, note first that~$F(\tilde{s}_2)=\dtw(s_1,\tilde{s}_2)^2$.
  Let~$x$ be a condensed string.
  If~$|x| < \ell_1$, then~$\dtw(s_1,x)^2 > 0$ and
  $\dtw(s_2,x)^2 \ge \dtw(s_2,\tilde{s}_1)^2 \ge \dtw(\tilde{s}_2,\tilde{s}_1)^2 = \dtw(\tilde{s}_2,s_1)^2$.
  Thus, $F(x) > F(\tilde{s}_2)$.
  Similarly, if~$|x| > \ell_2$, then~$\dtw(s_1,x)^2 \ge \dtw(s_1,\tilde{s}_2)^2$, $\dtw(s_2,x)^2 > 0$, and $F(x) > F(\tilde{s}_2)$.
  If~$\ell_1\le |x| < \ell_2$, then $\dtw(s_1,x)^2 + \dtw(s_2,x)^2 \ge \dtw(s_1,\tilde{s}_2)^2$, and thus~$F(x) \ge F(\tilde{s}_2)$.
\end{proof}

For three input strings, we show linear-time solvability if all strings begin with the same symbol and end with the same symbol.

\begin{theorem}\label{thm:k=3}
  Let~$s_1,s_2,s_3$ be binary strings with $s_1[1]=s_2[1]=s_3[1]$ and~$s_1[-1]=s_2[-1]=s_3[-1]$.
  A condensed mean of~$s_1,s_2,s_3$ can be computed in linear time.
\end{theorem}

\begin{proof}
  We first determine the condensations and block sizes of~$s_1,s_2$, and $s_3$ in linear time.
  Let~$\ell_i\coloneqq|\tilde{s}_i|$, for $i\in[3]$, and assume $\ell_1 \le \ell_2 \le \ell_3$.
  Note that every mean starts with~$s_1[1]$ and ends with~$s_1[-1]$. To see this, consider any string~$x$ with~$x[1]\neq s_1[1]$ (or $x[-1]\neq s_1[-1]$) and observe that either removing the first (or last) symbol or adding $s_1[1]$ to the front (or $s_1[-1]$ to the end) yields a better~$F$-value.
  Moreover, it is easy to see that every condensed mean has length at least~$\ell_2$ since increasing the length of any shorter condensed string by two increases the dtw~distance to~$s_1$ by at most one (\Cref{lem:dtwdist}) and decreases the dtw~distances to~$s_2$ and~$s_3$ by at least one (\Cref{thm:binDTW}).
  
  Note that a mean could be even longer than~$\ell_2$ since further increasing the length by two increases the dtw~distance to~$s_1$ and~$s_2$ by at most one and could possibly decrease the dtw~distance to~$s_3$ by at least two (if a misaligned block of size at least two can be saved). In fact, we can determine an optimal mean length in $O(\ell_3)$ time by greedily computing the maximum number $\rho$ of 1-separated (that is, non-neighboring) blocks of size one among~$s_3^{(2)},\ldots,s_3^{\ell_3-1}$. Then there is a mean of length~$\ell_3-2\rho$ (that is, exactly $\rho$ size-1 blocks of~$s_3$ are misaligned).
  Clearly, any longer condensed string has a larger $F$-value and every shorter condensed string has at least the same~$F$-value.
\end{proof}

We strongly conjecture that similar but more technical arguments can be used to obtain a linear-time algorithm for three arbitrary input strings. For more than three strings, however, it is not clear how to achieve linear time, since the mean length cannot be greedily determined.

\section{Empirical Evaluation}
We conducted some experiments to empirically evaluate our algorithms and to observe structural characteristics of binary means.
In \Cref{sec:alg-exp} we compare the running times of our $O(k(n+m(m-\mu)))$-time algorithm (\Cref{thm:main}~(ii)) with the standard $O(kn^2)$-time dynamic programming approach~\cite{BFFJNS19} described in the beginning of \Cref{sec:main}.
We implemented both algorithms in Python.\footnote{Source code available at \url{www.akt.tu-berlin.de/menue/software/}.}
Note that we did not test the $O(kn^{1.87})$-time algorithm since it uses another blackbox algorithm (which has not been implemented so far) in order to solve \MSS. However, we believe that it is anyway slower in practice.
In \Cref{sec:struct}, we empirically investigate structural properties of binary condensed means such as the length and the starting symbol (note that these two characteristics completely define the mean).
All computations have been done on an Intel i7 QuadCore~(4.0 GHz).

For our experiments we used the CASAS human activity datasets\footnote{Available at \url{casas.wsu.edu/datasets/}.}~\cite{CCTK13} as well as some randomly generated data.
The data in the CASAS datasets are generated from sensors which detect (timestamped) changes in the environment (for example, a door being opened/closed) and have previously been used in the context of binary dtw computation~\cite{MCAHM16}.
We used the datasets HH101--HH130 and sampled from them to obtain input strings of different lengths and sparsities (for a binary string~$s$, we define the \emph{sparsity} as $|\tilde s| / |s|$).
For the random data, the sparsity value was used as the probability that the next symbol in the string will be different from the last one (hence, the actual sparsities are not necessarily exactly the sparsities given but very close to those).

\subsection{Running Time Comparison}\label{sec:alg-exp}
\begin{figure}[ht]
  \centering
  \includegraphics[scale=0.365]{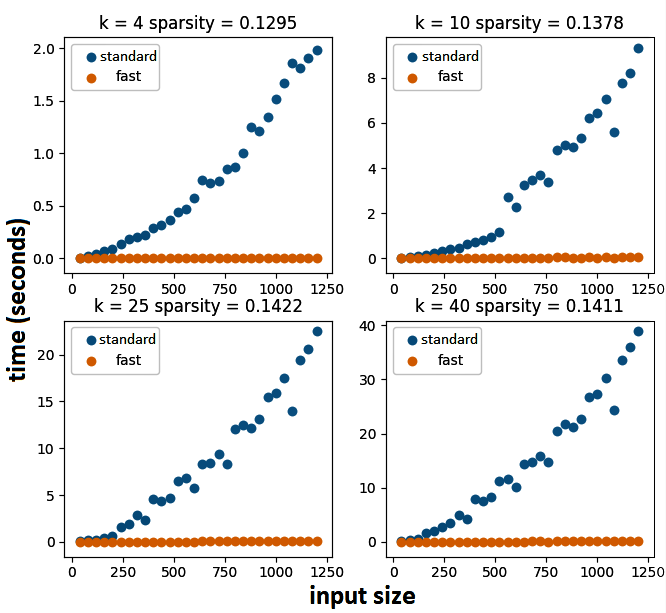}
  \includegraphics[scale=0.365]{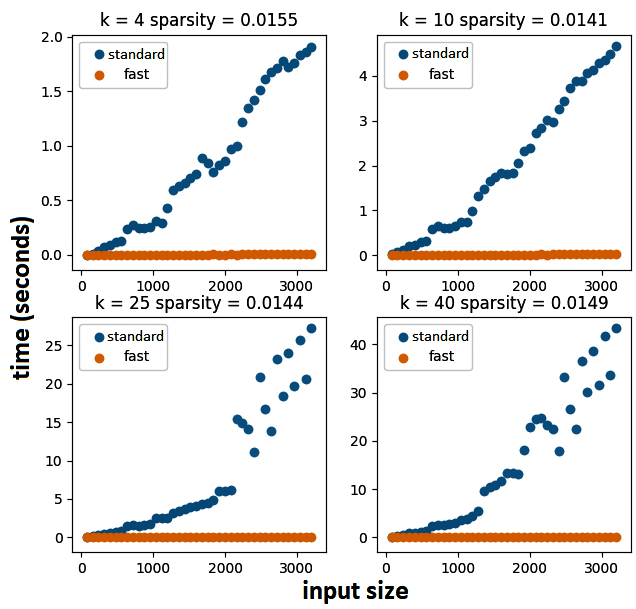}
  \captionof{figure}{Running times of the standard and the fast algorithm on sparse data (sensor D002 in dataset HH101) in 10-minute intervals (left) and 1-minute intervals (right).}
  \label{fig9}
\end{figure}

\begin{figure}[t]
  \centering
  \includegraphics[scale=0.37]{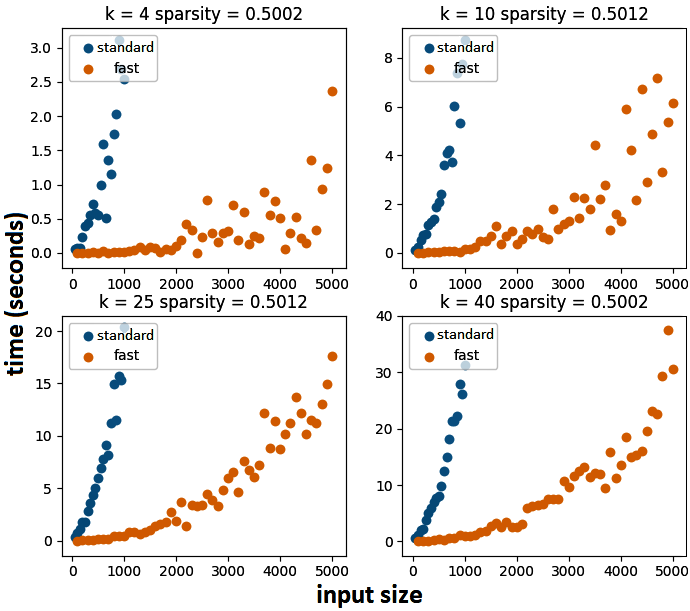}
  \includegraphics[scale=0.375]{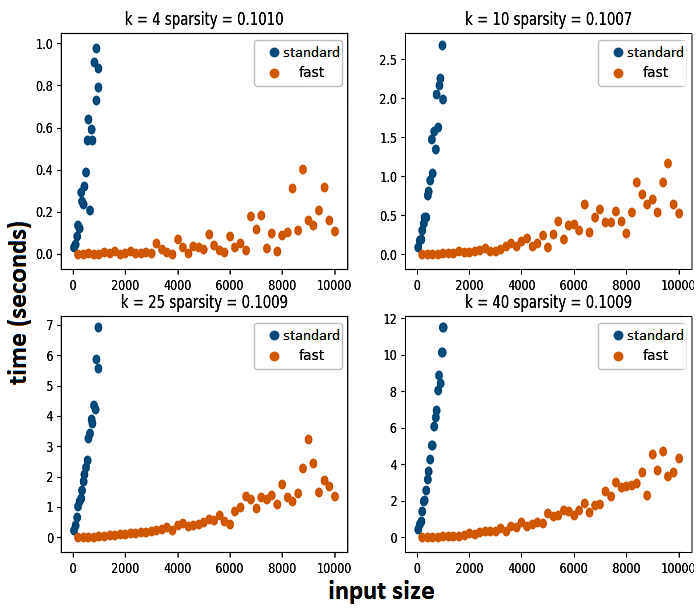}
  \captionof{figure}{Running times of the standard algorithm for $n \le 1000$ and of the fast algorithm for $n\le 5000$ on dense random data (left) and for $n\le 10,000$ on sparse random data (right).}
  \label{fig11}
\end{figure}

To examine the speedup provided by our algorithm, we compare it with the standard $O(kn^2)$-time dynamic programming algorithm on (very) sparse real-world data (sparsity $\approx 0.1$ and $\approx 0.01$) and on sparse (sparsity $\approx 0.1$) and dense (sparsity $\approx 0.5$) random data, both for various values of $k$.
\Cref{fig9} shows the running times on real-world data.
For sparsity $\approx 0.1$, our algorithm is around 250~times faster than the standard algorithm and for sparsity $\approx 0.01$ it is around 350~times faster.
\Cref{fig11} shows the running times of the algorithms on larger random data. For sparsity $\approx 0.1$, our algorithm is still twice as fast for $n=10,000$ as the standard algorithm for $n=1000$.
These results clearly show that our algorithm is valuable in practice.

\subsection{Structural Observations}\label{sec:struct}
We also studied the typical shape of binary condensed means.
The questions of interest are ``What is the typical length of a condensed mean?'' and ``What is the first symbol of a condensed mean?''. Since the answers to these two questions completely determine a condensed mean, we investigated whether they can be easily determined from the inputs. 

To answer the question regarding the mean length, we tested how much the actual mean length differs from the median condensed length. Recall that by \cref{lem:bounds} we know that every condensed mean has length at least $\mu-2$, where $\mu$ is the median input condensation length. We call this lower bound the median condensed length.
We used our algorithm (\Cref{thm:main}~(ii)) to compute condensed means
on random data with sparsities $0.01,\ldots,1.0$, $k=1,\ldots,60$ and $n \le 400$.
\Cref{fig13} clearly shows that on dense data (sparsity $ > 0.5$), the difference between the mean length and the median condensed length is rarely more than one.
This can be explained by the fact that for dense strings all blocks are usually small such that
there is no gain in making the mean longer than the median condensed length.
We remark that a difference of one appears quite often which might be caused by different starting or ending symbols of the inputs.
In general, for dense data the mean length almost always is at most the median condensed length plus one, whereas for sparser data the mean can become longer.
As regards the dependency of the mean length on the number~$k$ of inputs, it can be observed
that, for sparse data (sparsity $< 0.5$), the mean length differs even more for larger~$k$.
This may be possible because more input strings increase the probability that there is one input string with long condensation length and large block sizes. For dense inputs, there seems to be no real dependence on~$k$.

\begin{figure}[t]
  \center
  \includegraphics[scale=0.3]{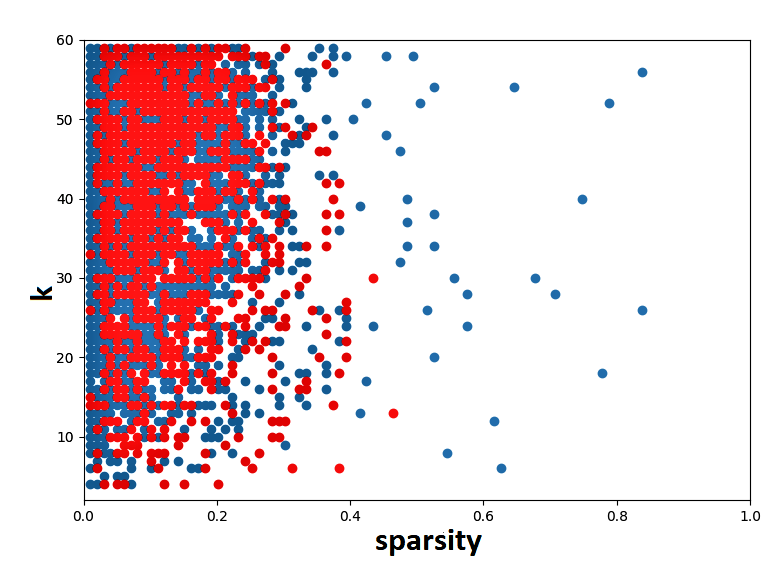}
  \caption{Difference between median condensed length and calculated mean length depending on sparsity and number of input strings. For every pair~$(\sigma,k)\in\{0.01,\ldots,1.0\}\times[60]$, we calculated one mean for $k$ strings with sparsity~$\sigma$. No dot means that the median condensed length and the mean length did not differ by more than one. A blue (dark gray) dot means they differed slightly (difference between two and four) and a red (light gray) dot means they differed by at least five.}
  \label{fig13}
\end{figure}

To answer the question regarding the first symbol of a mean, we tested on random data with different $k$~values and different sparsities ($n \le 500$), how the starting symbol of the mean depends on the starting symbols or blocks of the input strings. First, we tested how often the starting symbol of the mean equals the majority of starting symbols of the input strings (see \Cref{tbl1}).
Then, we also summed up the lengths of all starting 1-blocks and the lengths of all starting 0-blocks and checked how often the mean starts with the symbol corresponding to the larger of those two sums (see \Cref{tbl2}). Overall, the starting symbol of the mean matches the majority of starting symbols or blocks of the input strings in most cases ($\approx$ 70--90\%, increasing with higher sparsity).
For low sparsities, however, taking the length of starting blocks into account seems to yield less matches. This might be due to large outlier starting blocks (note that this effect is even worse for larger~$k$).
\begin{table}[t]
  \centering
\caption{Frequency (over 1000 runs) of the first symbol of the mean also being the first symbol in the majority of input strings.}
\begin{tabular}{c|cccccc}
$k$/sparsity & 0.05 & 0.1 & 0.2 & 0.5 & 0.8 & 1\\
\hline
5 & 76\% & 79\% & 82\% & 82\% & 82\% & 80\%\\
15 & 75\% & 81\% & 82\% & 83\% & 85\% & 85\%\\
40 & 82\% & 84\% & 88\% & 87\% & 91\% & 97\%
\end{tabular}
\label{tbl1}
\end{table}

\begin{table}[t]
  \centering
\caption{Frequency (over 1000 runs) of the first symbol of the mean also being the majority of symbols throughout the first blocks of input strings.}
\begin{tabular}{c|cccccc}
$k$/sparsity & 0.05 & 0.1 & 0.2 & 0.5 & 0.8 & 1\\
\hline
5 & 69\% & 73\% & 75\% & 83\% & 85\% & 80\% \\
15 & 67\% & 73\% & 75\% & 82\% & 88\% & 85\% \\
40 & 66\% & 70\% & 74\% & 81\% & 91\% & 97\%
\end{tabular}
\label{tbl2}
\end{table}

To sum up the above empirical observations, we conclude that a condensed binary mean typically has
a length close to the median condensed length and starts with the majority symbol among the starting symbols in the inputs.

\section{Conclusion}
In this work we made progress in understanding and efficiently computing binary means of binary strings with respect to the dtw~distance.
First, we proved tight lower and upper bounds on the length of a binary (condensed) mean which we then used to obtain
fast polynomial-time algorithms to compute binary means by solving a certain number problem efficiently.
We also obtained linear-time algorithms for~$k\le 3$ input strings.
Moreover, we empirically showed that the actual mean length is often very close to the proven lower bound.

As regards future research challenges, it would be interesting to further improve the running time with respect to the maximum input string length~$n$. This could be achieved by finding faster algorithms for our ``helper problem'' 
\textsc{Min 1-Separated Sum} (\MSS).
Can one solve \BDTW in linear time for every constant~$k$ (that is, $f(k)\cdot n$ time for some function~$f$)?
Also, finding improved algorithms for the weighted version or the center version (see Section~\ref{sec:main}) might be of interest.

\newpage

\bibliography{ref}

\begin{thebibliography}{20}
\providecommand{\natexlab}[1]{#1}
\providecommand{\url}[1]{\texttt{#1}}
\expandafter\ifx\csname urlstyle\endcsname\relax
  \providecommand{\doi}[1]{doi: #1}\else
  \providecommand{\doi}{doi: \begingroup \urlstyle{rm}\Url}\fi

\bibitem[{Abboud} et~al.(2015){Abboud}, {Backurs}, and {Williams}]{ABW15}
A.~{Abboud}, A.~{Backurs}, and V.~V. {Williams}.
\newblock Tight hardness results for {LCS} and other sequence similarity
  measures.
\newblock In \emph{2015 IEEE 56th Annual Symposium on Foundations of Computer
  Science (FOCS~'15)}, pages 59--78, 2015.

\bibitem[Brill et~al.(2019)Brill, Fluschnik, Froese, Jain, Niedermeier, and
  Schultz]{BFFJNS19}
M.~Brill, T.~Fluschnik, V.~Froese, B.~J. Jain, R.~Niedermeier, and D.~Schultz.
\newblock Exact mean computation in dynamic time warping spaces.
\newblock \emph{Data Mining and Knowledge Discovery}, 33\penalty0 (1):\penalty0
  252--291, 2019.

\bibitem[Bringmann and K{\"{u}}nnemann(2015)]{BK15}
K.~Bringmann and M.~K{\"{u}}nnemann.
\newblock Quadratic conditional lower bounds for string problems and dynamic
  time warping.
\newblock In \emph{2015 {IEEE} 56th Annual Symposium on Foundations of Computer
  Science (FOCS~'15)}, pages 79--97, 2015.

\bibitem[Buchin et~al.(2019)Buchin, Driemel, and Struijs]{BDS19}
K.~Buchin, A.~Driemel, and M.~Struijs.
\newblock On the hardness of computing an average curve.
\newblock \emph{CoRR}, abs/1902.08053, 2019.
\newblock Preprint appeared at the \emph{35th European Workshop on
  Computational Geometry (EuroCG~'19)}.

\bibitem[Bulteau et~al.(2014)Bulteau, H{\"{u}}ffner, Komusiewicz, and
  Niedermeier]{BHKN14}
L.~Bulteau, F.~H{\"{u}}ffner, C.~Komusiewicz, and R.~Niedermeier.
\newblock Multivariate algorithmics for {NP}-hard string problems.
\newblock \emph{Bulletin of the {EATCS}}, 114, 2014.

\bibitem[Bulteau et~al.(2018)Bulteau, Froese, and Niedermeier]{BFN18}
L.~Bulteau, V.~Froese, and R.~Niedermeier.
\newblock Tight hardness results for consensus problems on circular strings and
  time series.
\newblock \emph{CoRR}, abs/1804.02854, 2018.

\bibitem[Chen and Wang(2011)]{CW11}
Z.~Chen and L.~Wang.
\newblock Fast exact algorithms for the closest string and substring problems
  with application to the planted ($\ell$, $d$)-motif model.
\newblock \emph{{IEEE/ACM} Transactions Computational Biology and
  Bioinformatics}, 8\penalty0 (5):\penalty0 1400--1410, 2011.

\bibitem[Chen et~al.(2012)Chen, Ma, and Wang]{CMW12}
Z.~Chen, B.~Ma, and L.~Wang.
\newblock A three-string approach to the closest string problem.
\newblock \emph{Journal of Computer and System Sciences}, 78\penalty0
  (1):\penalty0 164--178, 2012.

\bibitem[Chen et~al.(2016)Chen, Ma, and Wang]{CMW16}
Z.~Chen, B.~Ma, and L.~Wang.
\newblock Randomized fixed-parameter algorithms for the closest string problem.
\newblock \emph{Algorithmica}, 74\penalty0 (1):\penalty0 466--484, 2016.

\bibitem[Cook et~al.(2013)Cook, Crandall, Thomas, and Krishnan]{CCTK13}
D.~Cook, A.~Crandall, B.~Thomas, and N.~Krishnan.
\newblock Casas: A smart home in a box.
\newblock \emph{Computer}, 46, 2013.

\bibitem[Frances and Litman(1997)]{FL97}
M.~Frances and A.~Litman.
\newblock On covering problems of codes.
\newblock \emph{Theory of Computing Systems}, 30\penalty0 (2):\penalty0
  113--119, 1997.

\bibitem[Gold and Sharir(2018)]{GS18}
O.~Gold and M.~Sharir.
\newblock Dynamic time warping and geometric edit distance: Breaking the
  quadratic barrier.
\newblock \emph{{ACM} Transactions on Algorithms}, 14\penalty0 (4):\penalty0
  50:1--50:17, 2018.

\bibitem[Gramm et~al.(2003)Gramm, Niedermeier, and Rossmanith]{GNR03}
J.~Gramm, R.~Niedermeier, and P.~Rossmanith.
\newblock Fixed-parameter algorithms for closest string and related problems.
\newblock \emph{Algorithmica}, 37\penalty0 (1):\penalty0 25--42, 2003.

\bibitem[Kuszmaul(2019)]{Kuszmaul19}
W.~Kuszmaul.
\newblock Dynamic time warping in strongly subquadratic time: Algorithms for
  the low-distance regime and approximate evaluation.
\newblock In \emph{Proceedings of the 46th International Colloquium on
  Automata, Languages, and Programming (ICALP~'19)}, pages 80:1--80:15, 2019.

\bibitem[Li et~al.(2002)Li, Ma, and Wang]{LMW02}
M.~Li, B.~Ma, and L.~Wang.
\newblock On the closest string and substring problems.
\newblock \emph{Journal of the {ACM}}, 49\penalty0 (2):\penalty0 157--171,
  2002.

\bibitem[Mueen et~al.(2016)Mueen, Chavoshi, Abu-El-Rub, Hamooni, and
  Minnich]{MCAHM16}
A.~Mueen, N.~Chavoshi, N.~Abu-El-Rub, H.~Hamooni, and A.~Minnich.
\newblock {AWarp}: Fast warping distance for sparse time series.
\newblock In \emph{2016 IEEE 16th International Conference on Data Mining
  (ICDM~'16)}, pages 350--359, 2016.

\bibitem[Nishimura and Simjour(2012)]{NS12}
N.~Nishimura and N.~Simjour.
\newblock Enumerating neighbour and closest strings.
\newblock In \emph{7th International Symposium on Parameterized and Exact
  Computation, ({IPEC} '12)}, pages 252--263. Springer, 2012.

\bibitem[Sakoe and Chiba(1978)]{SC78}
H.~Sakoe and S.~Chiba.
\newblock Dynamic programming algorithm optimization for spoken word
  recognition.
\newblock \emph{IEEE Transactions on Acoustics, Speech, and Signal Processing},
  26\penalty0 (1):\penalty0 43--49, 1978.

\bibitem[Sharabiani et~al.(2018)Sharabiani, Darabi, Harford, Douzali, Karim,
  Johnson, and Chen]{SDHDKJC18}
A.~Sharabiani, H.~Darabi, S.~Harford, E.~Douzali, F.~Karim, H.~Johnson, and
  S.~Chen.
\newblock Asymptotic dynamic time warping calculation with utilizing value
  repetition.
\newblock \emph{Knowledge and Information Systems}, 57\penalty0 (2):\penalty0
  359--388, 2018.

\bibitem[Yuasa et~al.(2019)Yuasa, Chen, Ma, and Wang]{YCMW19}
S.~Yuasa, Z.~Chen, B.~Ma, and L.~Wang.
\newblock Designing and implementing algorithms for the closest string problem.
\newblock \emph{Theoretical Computer Science}, 786:\penalty0 32--43, 2019.

\end{thebibliography}

\appendix

\end{document}